\documentclass[sigconf]{acmart}

\usepackage{macros}

\newcommand{\agdahtmlurl}{http://dwarn.se/naturals/}

\newcommand{\formalizationref}[2]{\href{\agdahtmlurl#2}{\texttt{#1}}}

\acmConference[To appear in LiCS'24]{Logic in Computer Science}{July 8--11}{Tallinn, Estonia}

\setcopyright{acmlicensed}
\acmISBN{}
\acmDOI{}
\makeatletter
\def\@copyrightpermission{}
\makeatother

\author{Christian Sattler}
\orcid{0000-0001-6374-4427}
\affiliation{%
	\institution{Chalmers University of Technology \\ and University of Gothenburg}
	\city{Gothenburg}
	\country{Sweden}
	}
\email{sattler@chalmers.se}

\author{David Wärn}
\orcid{0000-0003-2071-4476}
\affiliation{%
	\institution{Chalmers University of Technology \\ and University of Gothenburg}
	\city{Gothenburg}
	\country{Sweden}
	}
\email{warnd@chalmers.se}

\title{Natural numbers from integers}

\begin{document}

\begin{abstract}
In homotopy type theory, a natural number type is freely generated by an element and an endomorphism.
Similarly, an integer type is freely generated by an element and an automorphism.
Using only dependent sums, identity types, extensional dependent products, and a type of two elements with large elimination, we construct a natural number type from an integer type.
As a corollary, homotopy type theory with only $\Sigma$, $\mathsf{Id}$, $\Pi$, and finite colimits with descent (and no universes) admits a natural number type.
This improves and simplifies a result by Rose.
\end{abstract}

\maketitle


\section{Introduction}

In standard set theory and other impredicative background theories, the set of natural numbers is easily constructed from some source of infinity.
As soon as we have a set $X$ with a ``zero'' element $z : X$ and an injective ``successor'' function $s : X \to X$ whose image does not contain $z$, we may carve out the natural numbers by taking the intersection of all subsets of $X$ closed under $z$ and $s$.
However, this reasoning is essentially \emph{impredicative}: for this definition to make sense (and yield the expected induction principle), we need the collection of all subsets of $X$ to form a set, the power set of $X$.

In predicative settings such as Martin-Löf dependent type theory, this reasoning is ill-founded.
Inductive type formers such as natural numbers cannot be reduced to non-inductive constructions: without impredicativity as a free source of induction principles, we are stuck.
Therefore, the natural numbers (and other inductive types) are usually assumed axiomatically.

Homotopy type theory is a version of Martin-Löf type theory with function extensionality, univalence for any assumed universes, and generalizations of inductive types called \emph{higher inductive types}.
An example of a non-recursive higher inductive type is the \emph{circle} $S^1$, freely generated by an element $b$ and an identification of $b$ with itself.
As proved in~\cite{kraus-von-raumer:path-hits}, the loop space of the circle (\ie, the type of self-identifications of its base point) has the universal properties of the \emph{integers}: it is freely generated by an element (given by the trivial self-identification) and an automorphism (given by composition with the generating loop).
Although its loop space has infinitely many elements, the type itself is generated using purely finite means: the higher inductive type of the circle does not have any recursive constructors.
Categorically, it arises from only finite limits and colimits.

The integers are certainly a source of infinity.
And in an impredicative setting, it would be reasonable to expect that a natural number type can be carved out from it in some way.
But in a predicative setting such as homotopy type theory, the problem of constructing the natural numbers from the integers is perplexingly challenging.
For example, there does not seem to be an easy way of even defining the subtype of non-negative integers: all the standard approaches use an inductively defined predicate.
The only possible form of induction, integer induction, is \emph{reversible}, in the sense that the induction step must be an equivalence rather than an arbitrary implication.
This is a severe restriction.

The question to settle is this: in a predicative theory such as homotopy type theory, can we construct the natural numbers from the integers?
To our knowledge, this question was first raised in a discussion between Egbert Rijke and Mike Shulman on the nForum~\cite{egbert:nat-from-int}.
(Instead of the integers, we may also assume non-recursive higher inductive types such as the circle and assume a univalent universe.)

To the surprise of many experts, Robert Rose settled this question positively in his PhD thesis~\cite{rose:naturals}.
His construction is quite complicated, reflected in the length of his thesis.
There were also some caveats: he needs two nested univalent universes (although the outer one may be replaced by large elimination).

Our contribution in this paper is to provide a new construction of the natural numbers from the integers that we believe to be much simpler.
Furthermore, our construction is more general: we do not need any univalent universes whatsoever.
Instead, we just rely on effectivity of finite coproducts (that the coprojections are disjoint embeddings; equivalently, that we have a two-element type that satisfies descent, a version of large elimination that computes up to equivalence).
In fact, we present two construction: a direct one in \cref{direct} and an indirect one in \cref{indirect}.

The key novel idea of our work is the following. We would like to say that a
natural number is a non-negative integer. So we would like to be able to define
the map $\sigma : \Z \to \two$ from integers to Booleans that tells whether an
integer is negative or non-negative. A priori it is unclear how to do this: our
only tool is integer induction, but $\sigma$ is not induced by any automorphism
of $\two$. Consider instead how to define $\sigma$ \emph{supposing} we have the
naturals $\N$ available: then we define an automorphism of the coproduct $\N
\sqcup \one \sqcup \N$, get an induced map $\sigma : \Z \to \N \sqcup \one
\sqcup \N$ (which happens to be an equivalence), and postcompose with the
natural map $\N \sqcup (\one \sqcup \N) \to \two$. In our setting, we cannot use $\N$
but we can simply replace $\N \sqcup \one \sqcup \N$ with another coproduct
$A \sqcup B \sqcup A$. This has an automorphism with the desired properties
whenever we have an equivalence $A \simeq B \sqcup A$.
In particular, one can define an equivalence $\Z \simeq \Z \sqcup \Z$,
decomposing $\Z$ into an even part and an odd part, directly
using integer induction.

The categorical analogue of our result is the following.
Suppose we have a locally cartesian closed (higher) category with finite coproducts that satisfy descent.
Then given an integer object, we have a natural number object.
In particular, any locally cartesian closed (higher) category with finite colimits that satisfy descent has a natural number object.
In the absence of an internal language result of our type theory in such higher categories (related conjectures are made in \cite{homotopy-theory-of-type-theories} with progress in \cite{internal-language-lex,infinity-type-theories}), we believe that our constructions are sufficiently abstract that they can be directly replayed in any concrete model for higher categories.%
\footnote{
We note that a similar-looking result is claimed in Rasekh~\cite{rasekh:naturals}.
However, this is in a setting with impredicativity, which as explained above simplifies the problem greatly.
Furthermore, the proof that the integers $\Omega S^1$ are 0-truncated is broken; the mistake is hidden in the proof of Lemma~1.2.1: $\mathsf{Eq}(g_1, g_2)$ is not generally a subspace of $\mathsf{Map}_{/X}(g_1, g_2)$.

Predicatively, it does not seem possible to directly show that the integers are 0-truncated.
Of course, once we have constructed the natural numbers, they are 0-truncated as usual and so are the integers via $\Z \simeq \N \sqcup \one \sqcup \N$.
See \cref{integers-set} for another perspective.
}

\section{Setting: a restricted type theory}

We work in dependent type theory with only a limited set of type formers as set out below.
We use $\equiv$ to denote judgmental equality and $\defeq$ to denote judgmental definitions.
We generally follow the homotopy type theory book~\cite{hottbook} for notational conventions.

We have the identity type $x =_A y$ for $x, y : A$ in the sense of Martin-Löf (no equality reflection or axiom K).
We have the unit type $\one$, dependent sum types $\sm{x:A} B(x)$, and dependent product types $\prd{x:A} B(x)$, all satisfying judgmental $\beta$- and $\eta$-laws.
Dependent products are extensional, meaning function extensionality holds (phrased using identity types).

We assume a \emph{two-element type} $\two$, freely generated by elements $0, 1 : \two$.
Its universal property says that for any type $C$ with elements $d_0, d_1 : C$, the type of maps $f : \two \to C$ with identifications $f(0) =_C d_0$ and $f(1) =_C d_1$ is contractible.
From this, we can derive a dependent elimination principle whose computation rule holds up to identity type.

\subsection{Descent}

We assume the two-element type satisfies \emph{descent}.
This encapsulates elimination into a univalent universe, allowing us to omit universes from our type theory.
The principle is as follows.
Given types $A_0$ and $A_1$, we have a family $C(x)$ over $x : \two$ together with equivalences $C(0) \simeq A_0$ and $C(1) \simeq A_1$.
Here, the notion of equivalence is defined as usual in homotopy type theory (using dependent sums and products).
This is a homotopically weakened version of large elimination.

\subsection{Basic consequences}

\paragraph{Elimination principle for two-element type}

From the universal property of the two-element type $\two$, we can derive a dependent elimination principle whose reduction rule holds up to identity type.
More precisely, given $C(x)$ for $x : \two$ with $d_0 : C(0)$ and $d_1 : C(1)$, we obtain $e(x) : C(x)$ for $x : \two$ together with identifications $e(0) =_C d_0$ and $e(1) =_C d_1$.
Moreover, the type of such $e$ is contractible (this uses function extensionality).

\paragraph{Empty type}

We may construct the empty type as
\[
\zero \defeq 0 =_\two 1
\rlap{.}\]
Its recursion principle follows from descent for the two-element type: given any type $A$, we have a family $C(x)$ over $x : \two$ with $C(0) \simeq \one$ and $C(1) \simeq A$.
From $y : 0 =_\two 1$, we then get $C(0) \simeq C(1)$, so $A \simeq \one$, meaning that $A$ is contractible.
From this, we obtain the elimination principle for $\zero$ as usual.

\paragraph{$\two$ is a set}

By descent, we have a family $C(x)$ over $x : \two$ with $C(0) \simeq \one$ and $C(1) \simeq \zero$.
The sum over $x : \two$ of $C(x)$ has an element at $x \equiv 0$.
By induction over $\two$ and $\zero$, every element in the sum is equal to it.
This shows that $\sm{x:\two} C(x)$ is contractible.
Since $C(0)$ is contractible, it follows that $0 =_\two 0$ is contractible.
A dual argument shows that $1 =_\two 1$ is contractible, making $\two$ a set (\ie, 0-truncated) by $\two$-induction.

\paragraph{Binary coproducts}

From descent for the two-element type, we can construct the \emph{binary coproduct type} $A_0 \sqcup A_1$ of types $A_0$ and $A_1$.
First, we obtain $C(x)$ over $x : \two$ with $C(0) \simeq A_0$ and $C(1) \simeq A_1$ from descent.
We then define
\[
A_0 \sqcup A_1 \defeq \sm{x:\two} C(x).
\]
The coprojections $\tau_i : A_i \to A_0 \sqcup A_1$ for $i \in \braces{0, 1}$ are defined by passing backward along the equivalence $C(i) \simeq A_i$.
The universal property of binary coproducts reduces, using dependent products, to the universal property of the two-element type.
Given functions $f_i : A_i \to D$ for $i \in \braces{0, 1}$, we write $[f_0, f_1] : A_0 \sqcup A_1 \to D$ for the map given by the universal property.
Equivalently, the dependent elimination principle for coproducts is justified by that principle for the two-element type.
Its reduction rules again hold up to identity type.
We use the notation $[-, -]$ also in this dependent case.

\paragraph{No confusion for binary coproducts}

The constructors for binary coproducts are disjoint embeddings.
Disjointness follows from the tautology $\neg (0 =_\two 1)$.
The embedding property reduces to contractibility of $0 =_\two 0$ and $1 =_\two 1$.
Alternatively, we may derive no-confusion from descent below.

\paragraph{Descent for binary coproducts}

Descent for the two-element type implies descent for binary coproducts.
This is the following principle, again encapsulating a homotopically weakened version of large elimination into a univalent universe.
Given families $B_i(a_i)$ over $a_i : A_i$ for $i \in \braces{0, 1}$, we have a family $D(y)$ over $y : A_0 \sqcup A_1$ with equivalences $D(\tau_i(a_i)) \simeq B_i(a_i)$.

The justification is a special case of stability of colimits with descent under pullback.
By descent for the two-element type, we have a family $D'(x)$ over $x : \two$ with for $ \in \braces{0, 1}$:
\[
v'_i : D'(i) \simeq \sm{a_i : A_i} B_i(a_i)
\rlap{.}\]
Recall $C(x)$ over $x : \two$ with $u_i : C(i) \simeq A_i$ for $i \in \braces{0, 1}$.
Elimination provides $f_x : D'(x) \to C(x)$ such that $f_i$ corresponds to the projection to $A_i$ under the equivalences $v'_i$ and $u'_i$ for $x : \two$.
We define
\[
D(x, y) \defeq f_x^{-1}(y)
\]
and note that
\[
D(i, a_i) \equiv f_i^{-1}(a_i) \simeq \fst^{-1}(a_i) \simeq B_i(a_i)
\rlap{.}\]
Alternatively, we can derive descent from no-confusion.

\paragraph{Finite coproducts with no confusion and descent}

We derive finite coproducts from binary coproducts and the empty type.
This extends to properties such as no-confusion or descent (note that descent for nullary coproducts is vacuous).
For a coproduct of (external) arity $k \in \N$, we denote the coprojections by $\tau_i$ where $0 \leq i < k$.

\subsection{Formalization}

The first approach in this paper has been formalized in Agda~\cite{agda}.
\footnote{At the point of writing, an HTML rendering of the formalization is available at \href{https://dwarn.se/naturals/index.html}{https://dwarn.se/naturals/index.html}.
This version may be updated in the future.}
Where appropriate, lemmas and theorem in our paper come with a reference to the corresponding part of the formalization.
Agda has many features beyond the setting of this paper, including universes
and general inductive types. The formalization avoids such features
and does not use any library. The only inductive types declared in the
formalization are the ones in our setting. It does not seem feasible to avoid
mentioning the first universe $\mathsf{Set}$ entirely when formalizing Agda,
since it is used for example for type polymorphism and to express
the judgement ``$A$ is a type'', but we restrict our uses of
$\mathsf{Set}$ to such places where it could in principle be
eliminated.

\paragraph{Isomorphisms and equivalences}

A fundamental result in any library of homotopy type theory is the fact that
every map with a two-sided inverse has contractible fibers, and we need this
result also in the present work.
The well-known proof of this result, reproduced for example in Section 10.4 of~\cite{egbert:intro},
is not entirely trivial, requiring some higher path algebra.
We present a new proof which we argue is simpler and easier to remember.
A similar argument appears in lectures notes by Mart\'in Escard\'o~\cite{escardo:intro}.

We start from two types $A$ and $B$ with maps $f : A \to B$ and $g : B \to A$ such that
$g(f(a)) = a$ for all $a : A$ and $f(g(b)) = b$ for all $b : B$.
We claim that the action on paths \[\ap_f : (a =_A a') \to (f(a) =_B f(a'))\] is
an isomorphism for all $a, a' : A$, i.e.\ it also has a left inverse and a right inverse.
We have a map \[\ap_g : (f(a) =_B f(a')) \to (g(f(a)) =_A g(f(a'))\]
and an isomorphism \[e : (g(f(a)) =_A g(f(a')) \to (a =_A a')\] since $g(f(a)) = a$ and
$g(f(a')) = a'$. We have \[e(\ap_g(\ap_f(p))) =_{a =_A a'} p\] for $p : a =_A a'$ by
path induction on $p$ and using that $e$ sends reflexivity to reflexivity.
Thus $\ap_f$ has a left inverse. By the same argument, swapping the roles of $f$ and $g$,
$\ap_g$ has a left inverse. Since $e$ is an isomorphism, $\ap_g$ has
a left inverse and a right inverse, so it is an isomorphism. Thus $\ap_f$ is also an isomorphism.
From this, we prove that $f$ has contractible fibers as follows.
Since $b = f(g(b))$ for $b : B$, it suffices to consider the fiber of $f$ over $f(a)$ for $a : A$,
that is $\sm{a' : A} f(a) =_B f(a')$. Since $\ap_f$ is an isomorphism, this is isomorphic to
$\sm{a' : A} a =_A a'$, a singleton, hence contractible.

\section{Natural numbers and integers}

A \emph{natural number algebra} is a type $A$ together with an element $z : A$ and an endofunction $s : A \to A$.
Note that this is an external notion (lacking universes, we may not quantify internally over types $A$ in our type theory).
We generally refer to a natural number algebra just by its underlying type.
We use $z$ and $s$ generically to refer to the structure components of a natural number algebra $A$.
The \emph{structure map} of $A$ is the map $\one \sqcup A \to A$ induced by $z$ and $s$.

Given natural number algebras $A$ and $B$, the type of \emph{algebra morphisms} from $A$ to $B$ consists of a function $f : A \to B$ with $f(z) =_B z$ and $f(s(a)) =_B s(f(a))$.
We call a natural number algebra $A$ \emph{initial} if this type is contractible for every $B$.
We then say $A$ is a \emph{natural number type}.
Every natural number algebra $A$ has an identity algebra morphism $\id_A$.
Algebra morphisms $f : A \to B$ and $g : B \to C$ admit a composition $g \cc f : A \to C$.
These operations satisfy unit and associativity laws (phrased using identity types).
We will not need any higher coherence conditions.

We also have displayed analogues of these notions.
Given a natural number $A$, a \emph{natural number algebra $B$ displayed over $A$} is a family $B(a)$ over $a : A$ with $z : B(z)$ and $s_a(b) : B(s(a))$ for $a : A$ and $b : B(a)$.
The type of \emph{sections} of $B$ consists of $b(a) : B(a)$ for $a : A$ together with $b(z) =_{B(z)} z$ and $b(s(a)) =_{B(s(a))} s(b(a))$ for $a : A$.
We say that $A$ has \emph{elimination} if every natural number algebra displayed over it has a section.
As is standard, one proves (externally):

\begin{lemma}%
\label{nat-up-vs-elim}
A natural number algebra is initial exactly if it has elimination.
\qed%
\end{lemma}

The \emph{total algebra} of $B$ is obtained by taking the dependent sum and has an evident projection morphism to $A$.
This is what we mean when we regard $B$ as an algebra.
Conversely, an algebra $T$ with a morphism $f : T \to A$ can be converted to displayed form over $B$ by taking fibers of $f$.

\medskip

The following notion features in our second approach.
A variation with a family of nullary constructors appears in our first approach.

\begin{definition}%
\label{stable}
A natural number algebra $A$ is \emph{stable} if its structure map $[z, s] : \one \sqcup A \to A$ is an equivalence.
\end{definition}

This means that $z : \one \to A$ and $s : A \to A$ form a colimiting cocone, \ie, that $A$ is \emph{non-recursively} freely generated by $z$ and $s$.
Lambek's lemma states that every initial natural number algebra is stable.

We call a natural number algebra $A$ an \emph{integer algebra} if its endofunction $s$ is an equivalence.
Note that this condition is a proposition.
This justifies defining morphisms of integer algebras as morphisms of the underlying natural number algebras.
An integer algebra $A$ is \emph{initial} if the type of integer algebra morphisms from $A$ to $B$ is contractible for any integer algebra $B$.
We then say that $A$ is an \emph{integer type} (more verbosely, a \emph{type of integers}).

We have a notion of displayed integer algebra analogous to the case of natural number algebras.
The type of sections of a displayed integer algebra is defined as the type of sections of the underlying displayed natural number algebra.
Analogous to the case of natural numbers, we have:

\begin{lemma}%
\label{integer-up-vs-elim}
An integer algebra is initial exactly if it has elimination.
\qed%
\end{lemma}

Our goal in the rest of this article is to construct a natural number type from an integer type.
We thus now make the standing assumption of an integer type $\Z$, with element denoted $Z : \Z$ and automorphism denoted $S : \Z \simeq \Z$ (to distinguish from the natural number algebras we will consider).
Note that any other integer type is equivalent to it (by universality).
It thus makes sense to speak of \emph{the} integer type $\Z$.

The rest of this paper is devoted to proving the following:

\begin{theorem}%
\label{naturals-from-integers}
Assume an integer type.
Then we have a natural number type.
\end{theorem}

We provide two separate proofs, a direct one in \cref{direct} and an indirect one in \cref{indirect}.

\section{An equivalence $\Z \simeq \Z \sqcup \Z$}

Our starting point for constructing the natural numbers is the following observation.

\begin{lemma}[\formalizationref{Doubling.halve-iso}{Doubling.html\#1151}]%
\label{splitting-equivalence}
We have $\Z \simeq \Z \sqcup \Z$.
\end{lemma}

\begin{proof}
We first define the operation of \emph{squaring} integer algebras.
Given an integer algebra $X \equiv (X, z, s)$, its square $\Sq(X)$ is the integer algebra $(X, z, s \cc s)$.
Note that $s \cc s$ is an equivalence since $s$ is.
This operation is functorial: it has an evident action on morphisms of integer algebras (we do not need any higher witnesses of functoriality).
Furthermore, the functorial action reflects equivalences: if $\Sq(f)$ is invertible, then so is $f$.

Next, for an integer algebra $X \equiv (X, z, s)$, we define the \emph{twisted rotation} integer algebra $\Tw(X)$ with carrier $X \sqcup X$, element $\tau_0(z)$, and automorphism $r$ mapping $\tau_0(x)$ to $\tau_1(x)$ and $\tau_1(x)$ to $\tau_0(s(x))$.
This uses the universal property of binary coproducts to define $r$.

The automorphism of $\Sq(\Tw(X))$ maps $\tau_0(x)$ to $\tau_0(s(x))$ and $\tau_1(x)$ to $\tau_1(s(x))$.
That is, on each component, it is just given by the original automorphism $s$.
In particular, $\tau_0$ forms an algebra morphism from $X$ to $\Sq(\Tw(X))$.

By initiality of $\Z$, we have an algebra map $\Z \to \Sq(\Z)$.
Let us call its underlying map $\double : \Z \to \Z$.
Note that
\[
\bracks{\double, S \cc \double} : \Z \sqcup \Z \to \Z
\]
forms an algebra map from $\Tw(\Z)$ to $\Z$.

Again by initiality of $\Z$, we have an algebra map $c : \Z \to \Tw(\Z)$.
To finish the proof, it suffices to show that its underlying function $\Z \to \Z \sqcup \Z$ is invertible.
By initiality of $\Z$, the algebra morphism composite $\bracks{\double, S \cc \double} \cc c$ is the identity, hence also on underlying functions.
It remains to show that the underlying function of the algebra endomorphism
\[
u \defeq c \cc \bracks{\double, S \cc \double}
\]
on $\Tw(\Z)$ is the identity.

Note that $\Sq(u)$ is an algebra endomorphism on $\Sq(\Tw(\Z))$.
By initiality of $\Z$, we have $\Sq(u) \cc \tau_0 = \tau_0$.
On underlying maps, this means $u \cc \tau_0 = \tau_0$.
It remains to check $u \cc \tau_1 = \tau_1$.
Given $x : \Z$, we calculate
\begin{align*}
u(\tau_1(x))
&=
c(S(\double(x)))
\\&=
r(c(\double(x)))
\\&=
r(u(\tau_0(x)))
\\&=
r(\tau_0(x))
\\&=
\tau_1(x).
\qedhere
\end{align*}
\end{proof}

\begin{remark}
We note an abstract perspective on \cref{splitting-equivalence}.
We have an endo-adjunction on the (higher) category of natural number algebras: the left adjoint is twisted rotation and the right adjoint is squaring.
Both functors preserve integer algebras, so the adjunction descends to the full subcategory of integer algebras.
As left adjoints preserve initial objects, the twisted rotation $\Tw(\Z)$ is again initial,
   so it is equivalent to $\Z$.

To render this reasoning in our type theory without universes, we represent (higher) categories as a mixed external-internal notion: external at object level and internal at morphism level and above.
That is, we have (external) \emph{sets} of objects (since those involve types in the case of algebras), but \emph{types} of morphisms.
As usual, the construction depends only on a finite approximation of the higher-categorical coherence tower, in this case just composition of morphisms and one level of coherence.
For the adjunction, we just need the equivalence between hom-types without naturality.
The given proof of \cref{splitting-equivalence} can be seen as the unravelling of the abstract perspective in this manner.

(This mixed reasoning can be formalized in two-level type theory~\cite{two-level-type-theory} using a variation of the notion of \emph{wild category} with an outer type of objects and inner types of morphisms.)
\end{remark}

\section{Approximating the integers via halves}

If we already had the natural numbers $\N$, we could describe the integers as being built out of two copies of the natural numbers, one for the positive and one for the negative half:
\begin{equation}%
\label{int-decomposition-nat}
\Z \simeq \N \sqcup \one \sqcup \N
\end{equation}
Conversely, we may use a decomposition with similar properties to tell us something about the integers, for example when an integer is positive or negative.
This is the motivation behind the following construction.

\begin{construction}[\formalizationref{Signs.shift-equiv}{Signs.html\#141}]%
\label{int-from-halves}
Consider types $A$ and $B$ with an equivalence $e : A \simeq B \sqcup A$.
Then we have an automorphism on $A \sqcup B \sqcup A$ given by reassociating the equivalence
\[
A \sqcup (B \sqcup A) \simeq (A \sqcup B) \sqcup A
\]
where we act using $e$ on the left component and using the inverse of $e$ on the right component.
Given also an element $b : B$, this forms an integer algebra structure on $A \sqcup B \sqcup A$.

In the above situation, initiality of $\Z$ provides us with an algebra morphism from $\Z$ to $A \sqcup B \sqcup A$.
Restricting along the underlying function, the ternary decomposition on the right induces a decomposition
\begin{equation}%
\label{int-decomposition}
\Z \simeq \Z^- \sqcup \Z^0 \sqcup \Z^+
\end{equation}
of $\Z$ into three parts (this uses effectiveness of coproducts -- that the coprojections are disjoint embeddings).
The automorphism $S$ of $\Z$ restricts to separate equivalences $S^- : \Z^- \simeq \Z^- \sqcup \Z^0$ and $S^+ : \Z^0 \sqcup \Z^+ \simeq \Z^+$ that combine to give $S$ via the above decomposition.
Furthermore, since $Z : \Z$ is sent to $\tau_1(b)$, we know that $Z$ lies in the middle component $\Z^0$.
\end{construction}

In the presence of the naturals, one may prove that the decomposition~\eqref{int-decomposition} in fact agrees with the decomposition~\eqref{int-decomposition-nat}.

\section{First approach: direct}%
\label{direct}

We now give a direct construction of an initial natural number algebra.
First we apply \cref{int-from-halves} to the equivalence $\Z \simeq \Z \sqcup \Z$ from \cref{splitting-equivalence} and the element $Z : \Z$ to obtain the decomposition~\eqref{int-decomposition}.
We write $M$ for $\Z^0 \sqcup \Z^+$.
By construction, for $x : \Z$, we have that $x$ lies in $M$ iff $S(x)$ lies in $\Z^+$,
so $S$ restricts to an equivalence $M \simeq \Z^+$.
We write $\iota_z : \Z^0 \to M$ for the coprojection $\tau_0$ and $s : M \to M$
for the composite embedding $M \simeq \Z^+ \hookrightarrow M$, so that
$[\iota_z, s] : \Z^0 \sqcup M \simeq M$ is an equivalence.
We write $\iota_M$ for the embedding $M \hookrightarrow \Z$. Note
that $s : M \to M$ `lies above' $S : \Z \simeq \Z$ in the sense that
$\iota_M \circ s = S \circ \iota_M$.

We will show that $M$ has a universal property close to that of the natural numbers:
it is freely generated by $\iota_z : \Z^0 \to M$ and $s : M \to M$.
Note also that we have an element of $\Z^0$, namely $Z$.
To construct $\N$ from here, we will use a simple rectification argument.

We first explain how to prove that $M$ has the stated universal property.
We essentially follow a well-known strategy for reducing natural number recursion (which
constructs \emph{functions} out of $\N$) to natural number induction
(which proves proposition-valued \emph{predicates} on $\N$)
by considering an appropriate notion of ``partially defined inductive function''.
To this end, we first need a notion of ordering on $\Z$.

\begin{lemma}[\formalizationref{Signs.agda}{Signs.html}]%
\label{lt-properties}
$\Z$ has a proposition-valued relation $<$ with the following properties:
\begin{parts}
\item\label{lt-properties:S-right}
if $x < y$, then $x < S(y)$,
\item\label{lt-properties:S}
if $S(x) < S(y)$, then $x < y$,
\item\label{lt-properties:Z}
if $x : M$ then we do not have $\iota_M(x) < Z$,
\item\label{lt-properties:lt-S}
$x < S(x)$ for all $x$.
\end{parts}
\end{lemma}

\begin{proof}
Using initiality of $\Z$, we can define subtraction on $\Z$ such that $x - Z = x$ and $x - S(y) = S^{-1}(x-y)$.
It can then be proven by integer induction that $S(x) - S(y) = x - y$ and
$x - x = Z$.
We take $x < y$ to mean that $x - y$ lies in $\Z^-$.
All the listed properties can be verified directly.
\end{proof}

We will suppress witnesses of the relation $<$, writing a dash in their place, relying on references to the previous lemma to fill them as required.
This is harmless as $<$ is valued in propositions.

We recall some preliminaries on fixpoints.
Given an endofunction $t$ on a type $X$, we write $\fix(t)$ for the type of \emph{fixpoints} of $t$:
\[
\fix(t) \defeq \sm{x : X} t(x) = x
\rlap{.}\]
The below ``rolling rule'' is useful for manipulating fixpoints.

\begin{lemma}[\formalizationref{RollingRule.agda}{RollingRule.html}]%
\label{rolling}
For $f : X \to Y$ and $g : Y \to X$, we have an equivalence
\[
\fix(g \cc f) \simeq \fix(f \cc g)
\rlap{.}\]
\end{lemma}

\begin{proof}
Both types arise from
\[
\sm{x : X}{y : Y} (f(x) = y) \times (g(y) = x)
\]
by contracting singletons.
\end{proof}
Note that the forward map ${\fix(g \cc f) \to \fix(f \cc g)}$ is $f$ applied to the first component, and the first component of the inverse map is similarly given by $g$.
This follows from unfolding the proof above.

\subsection{The universal property of $M$}

This subsection is devoted to proving the universal property of $M$:

\begin{proposition}[\formalizationref{M.M-ind}{M.html\#3387}]%
\label{M-universal-property}
Let $A$ be a family over $M$ together with
\begin{itemize}
\item $z_A(x) : A(\iota_z(x))$ for $x : \Z^0$,
\item $s_A(a) : A(s(x))$ for (implicit) $x : M$ with $a : A(x)$.
\end{itemize}
Then we have $g(x) : A(x)$ for $x : M$ such that:
\begin{itemize}
\item
$g(\iota_z(x)) = z_A(x)$ for $x : \Z^0$,
\item
$g(s(x)) = s_A(g(x))$ for $x : M$.
\end{itemize}
\end{proposition}

Throughout this subsection, we fix $A$ together with $z_A$ and $s_A$ as above.
For $x : M$ and $u : \Z$ we take $x <' u$ to mean $\iota_M(x)  < u $.
We define for $u : \Z$ a type
\[
\pfun(u) \defeq \prd{x:M} (x <' u) \to A(x)
\]
of partial sections of $A$ defined below $u$.
Intuitively, an element of $\pfun(u)$ is a section of $A$ defined on a finite prefix of $M$.
For $u : \Z$, we have a \emph{restriction} map
\[
\res_u : \pfun(S(u)) \to \pfun(u)
\]
using \cref{lt-properties:S-right} of \cref{lt-properties} and an \emph{extension} map
\[
\ext_u : \pfun(u) \to \pfun(S(u))
\]
given by a case distinction using the equivalence $\Z^0 \sqcup M \simeq M$:
\begin{itemize}
\item $\ext_u(f,\iota_z(x),-) \defeq z_A(x)$ for $x$ in $\Z^0$,
\item $\ext_u(f,s(x),-) = s_A(f(x,-))$ by \cref{lt-properties:S} of \cref{lt-properties}.
\end{itemize}

\begin{lemma}[\formalizationref{M.ext-res-eq-res-ext}{M.html\#1690}]%
\label{commute}
The operations $\res$ and $\ext$ commute: for $u : \Z$, we have
\[
\ext_u \cc \res_u = \res_{S(u)} \cc \ext_{S(u)}
\rlap{.}\]
\end{lemma}

\begin{proof}
For each $f : \pfun(S(u))$ and $x : M$ with $x <' u$, we have to prove an equality in $A(x)$.
We do a case distinction on $x$ using the equivalence $\Z^0 \sqcup M \simeq M$.
Each case is direct by unfolding definitions.
\end{proof}

Now we define for $u : \Z$ a type
\[
\indfun(u) \defeq \fix(\res_u \cc \ext_u)
\]
of partial sections $f$ of $A$ defined below $u$ together with a witness that $f$ is \emph{inductive}.
To understand this terminology, note that $\res_u(\ext_u(f)) = f$ unfolds to the recursive equation
\[
f(\iota_z(x),-) = z_A(x)
\]
for $x : \Z^0$ with $\iota_z(x) <' u$ and
\[
f(s(x),-) = s_A(f(x),-)
\]
for $x : M$
with $s(x) <' u$.
We defined $\indfun$ in this more compact way in order to obtain a
simple proof of the following result.
\begin{lemma}[\formalizationref{M.M-ind-aux}{M.html\#3302}]%
For all $u : \Z$, we have an element $f(u)$ of $\indfun(u)$.
\end{lemma}

\begin{proof}
By integer induction. We trivially have an element of $\indfun(Z)$, since
there is no $x : M$ with $x <' Z$.
Moreover, we have
\begin{align*}
\indfun(S(u)) &\simeq \fix(\res_{S(u)} \cc \ext_{S(u)})  \\
	          &\simeq \fix(\ext_u \cc \res_u)\\
			  &\simeq \fix(\res_u \cc \ext_u)\\
			  &\simeq \indfun(u)
\end{align*}
by \cref{commute} and the rolling rule.
\end{proof}

In fact, one can strengthen the above result to the claim that $\indfun(u)$ is contractible -- so in particular the proof only uses integer induction for propositions -- but we will not need this strengthening.

\begin{proof}[Proof of \cref{M-universal-property}]
We define $g(x)$ as the evaluation of $f(S(x))$ at $x$, using \cref{lt-properties:lt-S} of \cref{lt-properties}.
(We write $x$ instead of $\iota_M(x)$ in order to declutter notation.)
The first equation follows from the fact that $f(S(x))$ is inductive.
The second equation follows from $f(S(x)) = \ext_x(f(x))$, which
we have by construction of $f$
and the fact that the rolling map
\[\fix(\res_x \cc \ext_x) \to \fix(\ext_x \cc \res_x)\]
is
\[
\ext_x : \indfun(x) \to \indfun(S(x))
\]
on first components.
\end{proof}

\subsection{Defining the natural numbers}

\begin{lemma}[\formalizationref{Naturals.agda}{Naturals.html}]
\label{rectify-nat}
Let $Y$ be a type with an element $z : Y$.
Let $X$ be freely generated by maps $\iota : Y \to X$ and $s : X \to X$.
Then we have a natural number type.
\end{lemma}
\begin{proof}
We define a self-map $r : X \to X$ by $r(\iota(y)) = \iota(z)$ and $r(s(x)) = s(r(x))$ using the universal property of $X$.
Let \[\N \defeq \sm{x : X} r(x) = x\] be the type of fixpoints of $r$.
By Lambek's lemma,%
\footnote{In the case at hand, we start from an equivalence $\Z^0 \sqcup M \simeq M$ and there is no need to invoke Lambek's lemma.}
the map $Y \sqcup X \to X$ is an equivalence, so in particular $\iota$ and $s$ are embeddings.
We have an equivalence \[e_\iota : (z = y) \simeq (r(\iota(y)) = \iota(y))\] since $\iota$ is an embedding,
and \[e_s : (r(x) = x)) \simeq (r(s(x)) = s(x))\] since $s$ is an embedding.
Thus we have an element $z_\N : \N$ given by $(z,e_\iota(\refl))$.
We also have an endomorphism $s_\N : \N \to \N$ given by $s_\N(x,p) = (s(x), e_s(p))$.

We claim that this makes $\N$ a natural number type.
Thus let $P$ be a type family over $\N$ with $z_P : P(z_\N)$ and \[s_P : \prd{n : \N} P(n) \to P(s_\N(n))\rlap{.}\]
We construct an element \[p : \prd{x : X} \prd{h : r(x) = x} P(x, h)\] using the universal property of $X$.
We first need to construct, for $y : Y$, an element \[\prd{h : r(\iota(y)) = \iota(y)} P(y, h)\rlap{,}\] or equivalently
\[\prd{h : z = y} P(y, e_\iota(h))\rlap{.}\] We can define this by path induction using $z_P$.
Then we have to construct, for $x : X$ and $f : \prd{h : r(x) = x} P(x,h)$, an element of \[\prd{h : r(s(x)) = s(x)} P(s(x),h)\rlap{.}\] or equivalently \[\prd{h : r(x) = x} P(s(x),e_s(h))\rlap{.}\]
Given $h : r(x) = x$, we simply use $s_P(x,h)(f(h))$.
It is direct to verify that this defines a section of $P$ as a displayed natural numbers algebra.
\end{proof}

\begin{proof}[Proof of \cref{naturals-from-integers}]
Apply \cref{rectify-nat} to $M$ with its universal property (\cref{M-universal-property}).
\end{proof}

\section{Second approach: indirect}%
\label{indirect}

The strategy of our second approach is more indirect and perhaps more in the spirit of Rose~\cite{rose:naturals}.
The key operation is \emph{stabilizing} natural number algebras: carving out a fragment (not generally a subtype) on which the structure map is invertible.
We achieve this by storing with each element a trace of previous elements that explains how the element is obtained by iterations of the structure map.
To construct a type of such traces for an element $x$, we need a function type valued in the carrier of the algebra, but domain varying in $x$.
If we had access to a univalent universe, we could make the choice of such a totally ordered domain type part of the data of the trace (we believe this is the essential part of Rose's construction).
Without univalence, the extra redundancy of this choice creates problems.
We eliminate this redundancy by forcing (the code for) the domain type to be uniquely determined by the values the trace function takes.
The technical complexity of the development results from encoding this \emph{very dependent type}~\cite{very-dependent-types}.
Lastly, to avoid a universe entirely, we custom-build a universe out of fixpoints of endofunctions on the integers that is closed under sufficient type formers.

\subsection{A custom universe}

We say that a family of types $B(x)$ for $x : A$ has \emph{binary coproducts} if, for $a_0, a_1 : A$, we have $t : A$ and maps $B(a_0) \to B(t)$ and $B(a_1) \to B(t)$ exhibiting $B(t)$ as a binary coproduct (equivalently, we have $t : A$ with an equivalence $B(t) \simeq B(a_0) \sqcup B(a_1)$).
We similarly define when the family has \emph{empty types (nullary coproducts)} and \emph{unit types}.
The intuition is that we see $A$ as a universe and $B$ as the associated universal family, sending a code in the universe to an actual type.
Note that we do not assume that the family is univalent.

\begin{lemma}%
\label{family-closed-under-finite-coproducts}
There is a family $\El$ over a type $U$ that has unit types and finite coproducts.
\end{lemma}

\begin{proof}
We take $U \defeq \Z \to \Z$ and $\El$ as taking fixpoints:
\[
\El(f) \defeq \sm{x:\Z} f(x) =_\Z x
\rlap{.}\]
The unit type is coded by the function constant on $Z$.
For finite coproducts, we make use of the equivalence $\Z \simeq \Z \sqcup \Z$ provided by \cref{splitting-equivalence}.
Under this equivalence, it suffices to exhibit codes as fixpoints of endofunctions on $\Z \sqcup \Z$.
Calculating the fixpoints of these endofunctions then makes use of the no-confusion property for binary coproducts.
\begin{itemize}
\item
The empty type is coded by the endofunction on $\Z \sqcup \Z$ swapping the two components.
The type of fixpoints of this endofunction is empty.
\item
The binary coproduct of $f, g : U$ is coded by the endofunction $f \sqcup g$ on $\Z \sqcup \Z$ that is separately $f$ on the left component and $g$ on the right component.
Its type of fixpoints is equivalent to the coproduct of the types of fixpoints of $f$ and $g$.
\qedhere
\end{itemize}
\end{proof}

\subsection{Counting structures}

\begin{definition}
The type of \emph{successor structures} from a type $C$ to a type $D$ is the record type (iterated dependent sum) with the following data:
\begin{itemize}
\item $\min : D$ and $\upp : C \to D$, together freely generating,
\item $\low : C \to D$ and $\max : D$, together freely generating,
\end{itemize}
\end{definition}

Note that free generation is expressed by a propositional type.
Equivalently, the maps
\begin{align*}
\bracks{\min, \upp} \co \one \sqcup C \to D
\rlap{,}\\
\bracks{\low, \max} \co C \sqcup \one \to D
\end{align*}
are invertible.
The induced equivalence
\[
\one \sqcup C \simeq D \simeq C \sqcup \one
\]
can be thought of as a successor relation on $C$ that misses a unique predecessor and successor from being an equivalence (as for example for a finite prefix of the natural numbers).
We could strengthen the definition to work with decidable total orders and require that the above equivalences lifts to
\[
\one \join C \simeq D \simeq C \join \one
\]
where $A \join B$ denotes the \emph{join} of orders $A$ and $B$.
However, we do not need this in our development.

\begin{definition}%
\label{counting-structure}
A \emph{counting structure} on a natural number algebra $A$ is a family $C$ over $A$ with $\neg C(z)$ and a successor structure $(\min_x, \upp_x, \low_x, \max_x)$ from $C(x)$ to $C(s(x))$ for $x : A$ such that $\min_{s(x)} \neq \max_{s(x)}$ in $C(s(s(x)))$ for $x : A$.
A \emph{counting algebra} is a natural number algebra equipped with a counting structure.
\end{definition}

The notion of counting structures is motivated by our desire to associate to each natural number a set of that cardinality.

\begin{lemma}%
\label{counting-algebra}
There is a counting algebra.
\end{lemma}

\begin{proof}
From \cref{family-closed-under-finite-coproducts}, we have a family $\El$ over a type $U$ with unit types and finite coproducts.
For our counting algebra, we take as carrier the type of tuples $(c, d, s)$ of $c, d : U$ with a successor structure $s = (\min, \upp, \low, \max)$ from $\El(c)$ to $\El(d)$.
The zero element has $c$ given by the code for the empty type and $d$ given by the code for the unit type, with essentially unique successor structure.
The successor of $(c, d, s)$ is given by $(d, e, t)$ where $e$ is a code for $\El(d) \sqcup \one$ and $t = (\min', \upp', \low', \max')$ has $\low' = \tau_0$ and $\max' = \tau_1(-)$ while $[\min', \upp']$ fills the below square of equivalences:
\[
\xymatrix{
  \one \sqcup (\El(c) \sqcup \one)
  \ar[r]^{\simeq}
  \ar[d]_{\one \sqcup \bracks{\low, \max}}^{\simeq}
&
  (\one \sqcup \El(c)) \sqcup \one
  \ar[d]^{\bracks{\min, \upp} \sqcup \one}_{\simeq}
\\
  \one \sqcup \El(d)
  \ar@{.>}[r]^-{\bracks{\min', \upp'}}
&
  \El(d) \sqcup \one
\rlap{.}}
\]
Note that $\min' \neq \max'$ by effectivity of binary coproducts.
\end{proof}

\begin{corollary}%
\label{map-from-counting-algebra}
Every natural number algebra receives a map from a counting algebra.
\end{corollary}

\begin{proof}
Take the product of the given natural number algebra with the counting algebra of \cref{counting-algebra}.
Since counting structures are contravariant in the natural number algebra, this product inherits a counting structure.
\end{proof}

\subsection{Stabilization}

Our only use of counting structures is to stabilize natural number algebras in \cref{map-from-stable} below.
The following statement encapsulates this use.
It provides a mechanism for annotating an algebra element with a ``recursive'' vector of algebra elements.

\begin{lemma}%
\label{cocounting-structure}
Every counting algebra $A$ admits the following:
\begin{itemize}
\item
a type $M(x)$ for $x : A$ with:
\begin{itemize}
\item
$M(z)$ is contractible,
\item
an equivalence $\pair : M(x) \times A \to M(s(x))$,
\end{itemize}
\item
$\last_x(m) : \one \sqcup A$ for $m : M(x)$ with identifications:
\begin{itemize}
\item
$\last_z(-) = \tau_0(\unitel)$,
\item
$\last_{s(x)}(\pair(m, y)) = \tau_1(y)$ for $m : M(x)$ and $y : A$,
\end{itemize}
\item
$\rest_x(m) : M(x)$ for $m : M(x)$ with identification
\[
\rest_{s(x)}(\pair(m, y)) =_{M(s(x))} \pair(\rest_x(m), \nextop_x(m))
\]
where $\nextop_x(m) \defeq [z, s](\last_x(m)) : A$.
\end{itemize}
\end{lemma}

\begin{proof}
Let $C$ denote the counting structure of $A$.
We define
\[
M(x) \defeq C(x) \to A
\rlap{.}\]
Note that $M(z)$ is contractible since $\neg C(z)$.
The equivalence
\[
\pair : (C(x) \to A) \times A \simeq (C(s(x)) \to A)
\]
is induced by freeness of $\low : C(x) \to C(s(x))$ and $\max : C(s(x))$.

For the other data, we use that $C(s(x))$ is freely generated by $\min$ and $\upp$.
The element $\last_x(m) : \one \sqcup A$ is defined by case distinction on $\max : C(s(x))$:
\begin{itemize}
\item
on $\min$, we return $\tau_0(\unitel)$,
\item
on $\upp(y)$ with $y : C(x)$, we return $\tau_1(m(y))$.
\end{itemize}

The element $\rest_x(m)(c) : A$ for $m : C(x) \to A$ and $c : C(x)$ is defined by case distinction on $\low(c) : C(s(x))$:
\begin{itemize}
\item
on $\min$, we return $z$,
\item
on $\upp(y)$ with $y : C(x)$, we return $m(y)$.
\end{itemize}

All the required identifications are direct.
\end{proof}

We now use the vectors provided by the previous lemma to annotate an algebra element with a trace witnessing that it is obtained recursively from the structure map.
Restricting to elements with such a trace stabilizes the natural number algebra.

\begin{lemma}%
\label{map-from-stable}
Every natural number algebra receives a map from a stable natural number algebra (see \cref{stable}).
\end{lemma}

\begin{proof}
Let $A$ denote the given natural number algebra.
Using \cref{map-from-counting-algebra}, we may reduce to the setting where $A$ comes equipped with a counting structure.
There, we have the structure given by \cref{cocounting-structure}.

We define a natural number algebra $B$ with underlying type the following record (iterated dependent sum):
\begin{itemize}
\item
$x : A$,
\item
$m : M(x)$,
\item
$q : x = \nextop_x(m)$.
\item
$p : m =_{M(x)} \rest_x(m)$,
\end{itemize}
This lies over $A$ via the first projection.
Given $x : A$, we write $B(x)$ for the record of the remaining three components.
It remains to construct an equivalence $\one \sqcup B \simeq B$ over $[z, s] : \one \sqcup A \to A$.
We construct this equivalence in reverse direction as a series of steps.

First, $B$ arises by contracting $k$ and $\alpha$ in the following record:
\begin{itemize}
\item
$x : A$,
\item
$k : \one \sqcup A$,
\item
$q : x =_A [z, s](k)$,
\item
$m : M(x)$,
\item
$\alpha : k =_{\one \sqcup A} \last_x(m)$,
\item
$p : m =_{M(x)} \rest_x(m)$.
\end{itemize}
Contracting $x$ with $q$, we obtain the equivalent record:
\begin{itemize}
\item
$k : \one \sqcup A$,
\item
$m : M(x)$,
\item
$\alpha : k = \last_x(m)$,
\item
$p : m = \rest_x(m)$
\end{itemize}
where the reverse direction sets $x \defeq [z, s](k)$ as required.
It remains to show that the record of the last three fields is equivalent to $B(k)$.
For this, we perform case distinction on $k$.

For $k = \tau_0(\unitel)$, we have $x = z$ and are equivalently left with:
\begin{itemize}
\item
$m : M(z)$,
\item
$\alpha : \tau_0(\unitel) = \last_z(m)$,
\item
$p : m = \rest_z(m)$.
\end{itemize}
Since $M(z)$ is contractible, this is equivalent to $\tau_0(\unitel) =_{\one \sqcup A} \tau_0(\unitel)$.
By no-confusion, $\tau_0$ is an embedding, so this is contractible.

For $k = \tau_1(y)$, we have $x = s(y)$ and are equivalently left with:
\begin{itemize}
\item
$m : M(s(y))$,
\item
$\alpha : \tau_1(y) = \last_{s(y)}(m)$,
\item
$p : m = \rest_{s(y)}(m)$.
\end{itemize}
Expanding the product $M(s(y)) \simeq M(y) \times A$, this is equivalent to:
\begin{itemize}
\item
$m' : M(y)$,
\item
$y' : A$,
\item
$\alpha : \tau_1(y) = \last_{s(y)}(\pair(m', y'))$,
\item
$p : \pair(m', z) = \rest_{s(y)}(\pair(m', y'))$.
\end{itemize}
This rewrites to:
\begin{itemize}
\item
$m' : M(y)$,
\item
$y' : A$,
\item
$\alpha : \tau_1(y) = \tau_1(y')$,
\item
$p : \pair(m', y') = \pair(\rest_y(m'), \nextop_y(m'))$.
\end{itemize}
Since $\tau_1$ is an embedding, we may contract $y'$ with $\alpha$:
\begin{itemize}
\item
$m' : M(y)$,
\item
$p : \pair(m', y) = \pair(\rest_y(m'), \nextop_y(m'))$.
\end{itemize}
Splitting $p$ into a pair of equalities, we recover $B(y)$.
\end{proof}

The strategy of the above proof is reminiscent of the proof of the rolling rule (\cref{rolling}).
In particular, the definition of $B$ is almost that of the fixpoints of an operation on $\sm{x:A} M(x)$.
However, the type of $p$ seems to resist this (it is an identification in $M(x)$, not $M(\nextop_x(m))$).
It is unclear to us how this analogy can be exploited.

\subsection{Defining the natural numbers}

\begin{lemma}%
\label{stable-embedding}
There is a stable natural number algebra that embeds into $\Z$.
\end{lemma}

\begin{proof}
By \cref{map-from-stable}, we have a stable natural number algebra $A$ with a morphism $f : A \to \Z$.
We now apply \cref{int-from-halves} with $A \defeq A$ and $B \defeq \one$.
The equivalence $A \simeq B \sqcup A$ is the structure map of the stable natural number algebra $A$.
Note that $B$ has a unique element.
The resulting integer algebra (with carrier $A \sqcup \one \sqcup A$) again lies over $\Z$: we send $\tau_0(a)$ to $S^{-1}(\inv(a))$, $\tau_1(\unitel)$ to $Z$, and $\tau_2(a)$ to $S(f(a))$.
Here, $\inv$ is the underlying map of the unique integer algebra morphism from $\Z$ to $\Z'$ where $\Z'$ has automorphism $S^{-1}$ instead of $S$.

By initiality of $\Z$, the algebra map
\[
\Z \to A \sqcup \one \sqcup A
\]
is a section of the algebra map
\[
A \sqcup \one \sqcup A \to \Z
\rlap{.}\]
By construction, the former map sends $\Z^0$ in the decomposition~\eqref{int-decomposition} to the middle component in $A \sqcup \one \sqcup A$, and in turn, that middle component is sent to $\Z^0$ by the latter map (specifically, to $Z$).
This exhibits $Z^0$ as a retract of $\one$, in particular it is contractible.
We may thus silently replace $Z^0$ by $\one$ in the obtained decomposition~\eqref{int-decomposition}.

This makes $\one \sqcup \Z^+$ with zero element $\tau_0(\unitel)$ and successor function $\tau_1 \cc S^+$ into a stable natural number algebra embedding into $\Z$.
\end{proof}

In the last step of the above lemma, we can equivalently directly use $\Z^+$ as the desired natural number algebra.
Denote the image of the zero in $\Z$ by $t$.
We then need to postcompose with the shifting map of $\Z$ that sends $t$ to $Z$ to obtain the algebra embedding from $\Z^+$ to $\Z$.

\begin{lemma}%
\label{stable-to-initial}
Let $A$ be a natural number algebra displayed over $\Z$ such that such that $A(Z)$ is contractible.
If $A$ is stable, then it is initial.
\end{lemma}

\begin{proof}
Using \cref{integer-up-vs-elim}, it suffices to show that $A$ has elimination.
In natural number algebras over $A$, having a section is a covariant structure.
That is, given a morphism $E' \to E$ of natural number algebras displayed over $A$, if $E'$ has a section, then so does $E$.
Using \cref{map-from-stable}, it thus suffices to construct a section for a \emph{stable} natural number algebra $E$ displayed over $A$.
In fact, we will show that $E$ is fiberwise contractible over $A$.
Since both $A$ and $E$ are stable, we have $E(Z, z) \simeq 1$ and $E(x, a) \simeq E(S(x), s(a))$ for $x : \Z$ and $a : A(x)$.

We define an integer algebra $Q$ displayed over $\Z$ with underlying family of propositions
\[
Q(x) \defeq \prd{a : A(x)} \isContr(E(x, a))
\rlap{.}\]
Note that $Q(Z)$ holds because $z : 1 \to A(Z)$ is an equivalence and $E(Z, z)$ is contractible.
We will check in the next paragraph that $Q(x)$ and $Q(S(x))$ are logically equivalent for $x : \Z$.
With that, initiality of $\Z$ gives a section $q$ of $Q$.
Then we have $q(x, a) : \isContr(E(x, a))$ for $a : A(x)$ as desired.

From $s : A(x) \to A(S(x))$ and $E(x, a) \simeq E(S(x), s(a))$, we have that $Q(S(x))$ implies $Q(x)$.
We check that $Q(x)$ implies $Q(S(x))$.
Given
\[
f : \prd{a : A(x)} \isContr(E(x, a))
\]
and $a' : A(S(x))$, we need to show that $E(S(x), a')$ is contractible.
We argue by cases using stability of $A$.
\begin{itemize}
\item
For $(S(x), a') = (Z, z)$, recall $E(Z, z)$ is contractible.
\item
For $(S(x), a') = (S(y), s(a))$ with $y : \Z$ and $a : A(y)$, note that $S(x) = S(y)$, hence $x = y$.
From $f$, we get that $E(y, a)$ contractible, hence so is $E(S(y), s(a))$.
\qedhere
\end{itemize}
\end{proof}

\begin{proof}[Proof of \cref{naturals-from-integers}]
By \cref{stable-embedding}, we have an embedding $X \hookrightarrow \Z$ of natural number algebras with $X$ stable.
Taking fibers, we obtain a stable natural number algebra $A$ displayed over $X$ with propositional fibers.
With $z : A(Z)$, this makes $A(Z)$ contractible.
We apply \cref{stable-to-initial} to $A$.
\end{proof}

\begin{remark}%
\label{integers-set}
\Cref{stable-embedding} implies that the point $Z : \Z$ is detachable.
In particular, the loop space at $Z$ is contractible.
Since $\Z$ is homogeneous, all its loop spaces are contractible.
This shows that $\Z$ is 0-truncated without first proving that the natural numbers are 0-truncated.
With \cref{stable-to-initial}, it implies that the natural numbers are 0-truncated.
\end{remark}

\bibliographystyle{ACM-Reference-Format}
\bibliography{naturals}

\end{document}